\documentclass[11pt]{article}

\usepackage[utf8]{inputenc} 
\usepackage{fullpage}
\usepackage{amsfonts,amsmath,amsthm,amssymb}
\usepackage{natbib}

\usepackage{graphicx}
\usepackage{subfigure}
\usepackage{tikz}
\usepackage{pgfplots}
\pgfplotsset{compat=newest}

\usepackage{hyperref}
\definecolor{mydarkblue}{rgb}{0,0.08,0.45}
\hypersetup{
   unicode=true,
   colorlinks=true,
   linkcolor=mydarkblue,
   citecolor=mydarkblue,
   filecolor=mydarkblue,
   urlcolor=mydarkblue,
   pdfview=FitH
}

\usepackage{booktabs, dsfont, enumitem, mathtools, microtype, nicefrac}

\usepackage{packages/algorithm}
\usepackage{packages/algorithmic}

\newtheorem{lemma}{Lemma}[section]
\newtheorem{theorem}[lemma]{Theorem}
\newtheorem{claim}[lemma]{Claim}

\newtheorem{definition}{Definition}

\newcommand{\ONE}{\mathds{1}}

\newcommand{\bfw}{\mathbf{w}}

\newcommand{\openLP}{\hrule height 0.8pt\rule{0pt}{1pt}} 
\newcommand{\closeLP}{\rule{0pt}{1pt}\hrule height 0.4pt\rule{0pt}{1pt}} 

\newcommand{\calC}{\mathcal{C}}
\newcommand{\calS}{\mathcal{S}}

\newcommand{\pr}{\mathrm{Pr}}

\newcommand{\In}{\text{In}}
\newcommand{\Out}{\text{Out}}

\newcommand{\given}{\;|\;}

\newcommand{\PlusSign}{\text{``$+$''}}
\newcommand{\MinusSign}{\text{``$-$''}}

\newcommand{\NoSign}{\text{``$\circ$''}}

\title{\texorpdfstring{Correlation Clustering with \\ Asymmetric Classification Errors}{Correlation Clustering with Asymmetric Classification Errors}\footnote{The conference version of this paper appeared in the proceedings of ICML 2020.}}
\author{
    Jafar Jafarov\thanks{Equal contribution. Jafar Jafarov and Yury Makarychev were supported in part by NSF CCF-1718820 and
NSF TRIPODS CCF-1934843.
Sanchit Kalhan and Konstantin Makarychev were supported in part by NSF TRIPODS CCF-1934931.} \\ University of Chicago 
 \and
    Sanchit Kalhan\footnotemark[2] \\ Northwestern University 
 \and
    Konstantin Makarychev\footnotemark[2] \\ Northwestern University 
 \and
    Yury Makarychev\footnotemark[2] \\ Toyota Technological Institute at Chicago
}
\date{}
\date{}
\begin{document}
\maketitle

\begin{abstract}
In the Correlation Clustering problem, we are given a weighted graph $G$
with its edges labelled as ``similar'' or ``dissimilar'' by a binary
classifier. The goal is to produce a clustering that minimizes the weight
of ``disagreements'': the sum of the weights of ``similar'' edges across
clusters and ``dissimilar'' edges within clusters. We study the correlation
clustering problem under the following assumption: Every ``similar'' edge
$e$ has weight $\bfw_e\in[\alpha \bfw, \bfw]$ and every ``dissimilar'' edge
$e$ has weight $\bfw_e\geq \alpha \bfw$ (where $\alpha\leq 1$ and $\bfw>0$
is a scaling parameter). We give a $(3 + 2 \log_e (1/\alpha))$
approximation algorithm for this problem. This assumption captures well the
scenario when classification errors are asymmetric. Additionally, we show
an asymptotically matching Linear Programming integrality gap of
$\Omega(\log 1/\alpha)$.
\end{abstract}

\section{Introduction}
In the Correlation Clustering problem, we are given a set of objects with pairwise similarity
information. Our aim is to partition these objects into clusters that match this information
as closely as possible. The pairwise information is represented as a weighted graph $G$ whose
edges are labelled as ``positive/similar'' and ``negative/dissimilar'' by a noisy binary classifier.
The goal is to find a clustering $\calC$ that minimizes the weight of edges disagreeing
with this clustering: A positive edge is in disagreement with $\calC$, if its endpoints belong to distinct clusters;
and a negative edge is in disagreement with $\calC$ if its endpoints belong to the same cluster. We call this objective the MinDisagree objective.
The MinDisagree objective has been
extensively studied in literature since it was introduced by~\citet*{BBC04}~(see e.g., \cite{CGW03, DEFI06, ACN08,pan2015, CMSY15}).
There are currently two standard models for Correlation Clustering which we will refer to as (1) Correlation Clustering on Complete Graphs
and (2) Correlation Clustering with Noisy Partial Information. In the former model, we assume that graph
$G$ is complete and all edge weights are the same i.e., $G$ is unweighted. In the latter model, we do not make any assumptions
on the graph $G$. Thus, edges can have arbitrary weights and some edges may be missing. These models are quite different
from the computational perspective. For the first model, \citet*{ACN08} gave a 2.5 approximation algorithm. This approximation
factor was later improved to 2.06  by Chawla, Makarychev, Schramm, and Yaroslavtsev [2015]. For the second model,
\citet*{CGW03} and \citet*{DEFI06} gave an $O(\log n)$ approximation algorithm, they
also showed that Correlation Clustering with Partial Noisy Information is as hard as
the Multicut problem and, hence, $O(\log n)$ is likely to be the best possible approximation for this problem.
In this paper, we show how to interpolate between these two models for Correlation Clustering.

We study the Correlation Clustering problem on complete graphs
with edge weights. In our model, the weights on the
edges are constrained such that the ratio of the lightest edge in the graph
to the heaviest positive edge is at least $\alpha \leq 1$. Thus, if $\bfw$ is the
weight of the heaviest positive edge in the graph, then each positive edge
has weight in $[\alpha \bfw, \bfw]$ and each negative edge has weight greater
than or equal to $\alpha \bfw$.
We argue that this model -- which we call Correlation Clustering with Asymmetric Classification Errors -- is more adept at capturing the subtleties in real world instances than the two standard models. Indeed, the assumptions made by the Correlation Clustering on Complete Graphs model are too strong, since rarely do real world instances have equal edge weights. In contrast, in the Correlation Clustering with Noisy Partial Information model we can have edge weights that are arbitrarily small or large, an assumption which is too weak.
In many real world instances, the edge weights lie in some range $[a, b]$ with $a,b > 0$. Our model captures a larger family of instances.

Furthermore, the nature of classification errors for objects that are
similar and objects that are dissimilar is quite different. In many cases, a
\textit{positive} edge $uv$ indicates that the classifier found some actual
evidence that $u$ and $v$ are similar; while a negative edge simply means
that the classifier could not find any such proof that $u$ and $v$ are
similar, it does not mean that the objects $u$ and $v$ are necessarily
dissimilar. In some other cases, a \textit{negative} edge $uv$ indicates that
the classifier found some evidence that $u$ and $v$ are dissimilar; while a
positive edge simply means that the classifier could not find any such proof.
We discuss several examples below. Note that in the former case, a positive
edge gives a substantially stronger signal than a negative edge and should
have a higher weight; in the latter, it is the other way around: a negative
edge gives a stronger signal than a positive edge and should have a higher
weight. We make this statement more precise in Section~\ref{sec:LML}.

The following examples show how the Correlation Clustering with Asymmetric
Classification Errors model can help in capturing real world instances.
Consider an example from the paper on Correlation Clustering
by~\citet*{pan2015}. In their experiments, \citet{pan2015} used several data
sets including \emph{dblp-2011} and \emph{ENWiki-2013}\footnote{These data
sets are published by~\cite{dataset1, dataset2, dataset3, dataset4}}. In the
graph~\emph{dblp-2011}, \emph{each vertex represents a scientist and two
vertices are connected with an edge if the corresponding authors have
co-authored an article}. Thus, a positive edge with weight $\bfw^+$ between
Alice and Bob in the Correlation Clustering instance indicates that Alice and
Bob are coauthors, which strongly suggests that Alice and Bob work in similar
areas of Computer Science. However, it is not true that all researchers
working in some area of computer science have co-authored papers with each
other. Thus, the negative edge that connects two scientists who do not have
an article together does not deserve to have the same weight as a positive
edge, and thus can be modeled as a negative edge with weight $\bfw^- <
\bfw^+$.

Similarly, the vertices of the graph \emph{ENWiki-2013} are Wikipedia pages.
Two pages are connected with an edge if there is a link from one page to
another. A link from one page to the other is a strong suggestion that the
two pages are related and hence can be connected with a positive edge of
weight $\bfw^+$, while it is not true that two similar Wikipedia pages
necessarily should have a link from one to the other. Thus, it would be
better to join such pages with a negative edge of weight $\bfw^- < \bfw^+$.

Consider now the multi-person tracking problem. The problem is modelled as a
Correlation Clustering or closely related Lifted Multicut
Problem~\cite{tang2016multi,tang2017multiple} on a graph, whose vertices are
people detections in video sequences. Two detections are connected with a
positive or negative edge depending on whether the detected people have
similar or dissimilar appearance (as well as some other information). In this
case, a negative edge $(u,v)$ is more informative since it signals that the
classifier has identified body parts that do not match in detections $u$ and
$v$ and thus the detected people are likely to be different (a positive edge
$(u, v)$ simply indicates that the classifier was not able to find
non-matching body parts).

The Correlation Clustering with Asymmetric Classification Errors model
captures the examples we discussed above. It is instructive to consider an
important special case where all positive edges have weight $\bfw^+$ and all
negative edges have weight $\bfw^-$ with $\bfw^+ \neq \bfw^-$. If we were to
use the state of the art algorithm for Correlation Clustering on Complete
Graphs on our instance for Correlation Clustering with Asymmetric
Classification Errors (by completely ignoring edge weights and looking at the
instance as an unweighted complete graph), we would get a
$\Theta(\max(\nicefrac{\bfw^+}{\bfw^-}, \nicefrac{\bfw^-}{\bfw^+}))$
approximation to the MinDisagree objective. While if we were to use the state
of the art algorithms for Correlation Clustering with Noisy Partial
Information on our instance, we would get a $O(\log n)$ approximation to the
MinDisagree objective.

\noindent \textbf{Our Contributions.} In this paper, we present an
approximation algorithm for Correlation Clustering with Asymmetric
Classification Errors. Our algorithm gives an approximation factor of $A = 3 + 2\log_e
\nicefrac{1}{\alpha}$. Consider the scenario discussed above where all
positive edges have weight $\bfw^+$ and all negative edges have weight
$\bfw^-$. If $\bfw^+ \geq \bfw^-$, our algorithm gets a $(3 + 2\log_e
\bfw^+/\bfw^-)$ approximation; if $\bfw^+ \leq \bfw^-$, our algorithm gets a
$3$-approximation.

\begin{definition}
Correlation Clustering with Asymmetric Classification Errors is a variant of
Correlation Clustering on a Complete Graph. We assume that the weight
$\bfw_e$ of each positive edge lies in $[\alpha \bfw, \bfw]$ and the weight
$\bfw_e$ of each negative edge lies in $[\alpha \bfw, \infty)$, where $\alpha
\in (0,1]$ and $\bfw > 0$.
\end{definition}

We note here that the assumption that the weight of positive edges is bounded from above is crucuial. Without this assumption (even if we require that negative weights are bounded from above and below), the LP gap is unbounded for every fixed $\alpha$ (this follows from the integrality gap example we present in Theorem~\ref{thm:intGap}).

The following is our main theorem.

\begin{theorem}\label{thm:main}
There exists a polynomial time $A = 3 + 2\log_e 1/\alpha$ approximation
algorithm for Correlation Clustering with Asymmetric Classification Errors.
\end{theorem}

We also study a natural extension of our model to the case of complete bipartite graphs. That is, the positive edges across the biparition have a weight between $[\alpha \bfw, \bfw]$ and the negative edges across the bipartition have a weight of at least $\alpha \bfw$. Note that the state-of-the-art approximation algorithm for Correlation Clustering on Unweighted Complete Bipartite Graphs has an approximation factor of $3$ (see ~\citet{CMSY15}).

\begin{theorem}\label{thm:mainBipartite}
There exists a polynomial time $A = 5 + 2\log_e 1/\alpha$ approximation
algorithm for Correlation Clustering with Asymmetric Classification Errors on complete bipartite graphs.
\end{theorem}

Our next result shows that this approximation ratio is likely best possible for LP-based algorithms.
We show this by exhibiting an instance of Correlation Clustering with
Asymmetric Classification Errors such that integrality gap for the natural LP
for Correlation Clustering on this instance is $\Omega(\log
\nicefrac{1}{\alpha})$.

\begin{theorem}\label{thm:intGap}
The natural Linear Programming relaxation for Correlation Clustering has an
integrality gap of $\Omega(\log \nicefrac{1}{\alpha})$ for instances of
Correlation Clustering with Asymmetric Classification Errors.
\end{theorem}

Moreover, we can show that if there is an $o(\log(1/\alpha))$-approximation algorithm whose running time is polynomial in both $n$ and $1/\alpha$, then there is an $o(\log n)-$approximation algorithm for the general weighted case
\footnote{The reduction to the general case works as follows. Consider an instance of Correlation Clustering with arbitrary weights. \emph{Guess} the heaviest edge $e$ that is in disagreement with the optimal clustering. Let $\bfw_e$ be its weight, and set
$\bfw = n^2 \bfw_e$, and $\alpha=1/n^4$. Then, assign new weights to all pairs of vertices in the graph. Keep the weights of all edges with weight in the range $[\alpha \bfw, \bfw]$. Set the weights of all edges with weight greater than $\bfw$ to $\bfw$ and the weights of all edges with weight less than $\alpha \bfw$ (including missing edges) to $\alpha \bfw$.}
(and also for the MultiCut problem). However, we do not know if there is an $o(\log(1/\alpha))-$approximation algorithm for the problem whose running time is polynomial in $n$ and \emph{exponential} in $1/\alpha$. The existence of such an algorithm does not imply that there is an $o(\log n)-$approximation algorithm for the general weighted case (as far as we know).

We show a similar integraplity gap result for the Correlation Clustering with Asymmetric Classification Errors on complete bipartite graphs problem.

\begin{theorem}\label{thm:intGapBipartite}
The natural Linear Programming relaxation for Correlation Clustering has an
integrality gap of $\Omega(\log \nicefrac{1}{\alpha})$ for instances of
Correlation Clustering with Asymmetric Classification Errors on complete bipartite graphs.
\end{theorem}

Throughout the paper, we denote the set of positive edges by $E^+$ and the
set of negative edges by $E^-$. We denote an instance of the Correlation
Clustering problem by $G=(V, E^+,E^-)$. We denote the weight of edge $e$ by
$\bfw_e$.

\subsection{Ground Truth Model} \label{sec:LML}
In this section, we formalize the connection between asymmetric classification errors and asymmetric edge weights. For simplicity, we assume that each positive edge has a weight of $\bfw^+$ and each negative edge has a weight of $\bfw^-$.
Consider
a probabilistic model in which edge labels are assigned by a noisy classifier. Let
$\calC^*=(C^*_1,\dots C^*_T)$ be the ground truth clustering of the vertex set $V$.
The classifier labels each edge within a cluster with a $\PlusSign$ edge
with probability $p^+$ and as a $\MinusSign$ edge with probability $1-p^+$;
it labels each edge with endpoints in distinct clusters as a $\MinusSign$ edge
with probability $q^-$ and as a $\PlusSign$ edge with probability $1-q^-$. Thus,
 $(1-p^+)$ and $(1-q^-)$ are the classification error probabilities.
We assume that all classification errors are independent.

We note that
similar models have been previously studied by \cite{BBC04, EW2009, MW10,
ACX13, MMV15-CC} and others. However, the standard assumption in such models
was that the error probabilities, $(1-p^+)$ and $(1-q^-)$, are less than a
half; that is, $p^+ > \nicefrac{1}{2}$ and $q^->\nicefrac{1}{2}$. Here, we
investigate two cases (i) when $p^+ < \nicefrac{1}{2} < q^-$ and (ii) when
$q^- < \nicefrac{1}{2} < p^+$. We assume that $p^{+}+q^{-}>1$, which means
that the classifier is more likely to connect similar objects with a
$\PlusSign$ than dissimilar objects or, equivalently, that the classifier is
more likely to connect dissimilar objects with a $\MinusSign$ than similar
objects. For instance, consider a classifier that looks for evidence that the
objects are similar: if it finds some evidence, it adds a positive edge;
otherwise, it adds a negative edge (as described in our examples
\emph{dblp-2011} and \emph{ENWiki-2013} in the Introduction). Say, the
classifier detects a similarity between two objects in the same ground truth
cluster with a probability of only $30\%$ and incorrectly detects similarity
between two objects in different ground truth clusters with a probability of
$10\%$. Then, it will add a \emph{negative} edge between two similar objects
with probability $70\%$! While this scenario is not captured by the standard
assumption, it is captured by case (i) (here, $p^+ = 0.3 < \nicefrac{1}{2} <
q^- = 0.9$ and $p^+ + q^- > 1$).

Consider a clustering $\calC$ of the vertices.
Denote the sets of positive edges and negative edges with
both endpoints in the same cluster by  $\In^+(\calC)$ and $\In^-(\calC)$, respectively, and
the sets of positive edges and negative edges with endpoints in different clusters by
$\Out^+(\calC)$ and $\Out^-(\calC)$, respectively.
Then, the log-likelihood function of the clustering $\calC$ is,
\begin{align*}
\ell(G; \calC) &=
\log\Big(
{\prod_{(u,v)\in \text{\In}^+(\calC)}}p^+ \times
\smashoperator{\prod_{(u,v)\in \In^-(\calC)}} (1-p^+) 
\times\smashoperator{\prod_{(u,v)\in \Out^+(\calC)}}(1-q^-) \times
{\prod_{(u,v)\in \Out^-(\calC)}}q^-
\Big)\\
&= \log\Big((p^+)^{|\In^+(\calC)|}(1-p^+)^{|\In^-(\calC)|}
\cdot (1-q^-)^{|\Out^+(\calC)|} (q^-)^{|\Out^-(\calC)|}\Big)\\
&= |\In^+(\calC)| \log p^+ + |\In^-(\calC)| \log (1-p^+)
+ |\Out^+(\calC)| \log (1-q^-) + |\Out^-(\calC)|\log q^-\\
&= \underbrace{\Big(|E^+| \log p^+ + |E^-| \log q^-\Big)}_{\text{constant expression}}
-
\underbrace{\Big(|\Out^+(\calC)| \log\frac{p^+}{1-q^-} + |\In^-(\calC)| 
\log \frac{q^-}{1-p^+}\Big)}_{\text{MinDisagree objective}}.
\end{align*}
Let $\bfw^+ = \log\frac{p^+}{1-q^-}$ and
$\bfw^- = \log \frac{q^-}{1-p^+}$. Then, the negative term --
$\Big(|\Out^+(\calC)| \log\frac{p^+}{1-q^-} + |\In^-(\calC)| \log \frac{q^-}{1-p^+}\Big)$
-- equals $\bfw^+ |\Out^+(\calC)| + \bfw^- |\In^-(\calC)|$.
Note that
$|\Out^+(\calC)|$ is the number of positive edges disagreeing with $\calC$ and
$|\In^-(\calC)|$ is the number of negative edges disagreeing with $\calC$.

Now observe that the first term in the expression above -- $\Big(|E^+| \log p^+ + |E^-| \log q^-\Big)$ --
does not depend on $\calC$. It only depends on the instance $G=(V, E^+, E^-)$.
Thus,
maximizing the log-likelihood function over $\calC$ is equivalent to minimizing
the following objective
$$\bfw^+ ({\text{\# disagreeing \PlusSign edges}}) + \bfw^- ({\text{\# disagreeing \MinusSign edges}}).$$

Note that we have $\bfw^+ > \bfw^-$ when $p^+ < \nicefrac{1}{2} < q^-$ (case
(i) above); in this case, a $\PlusSign$ edge gives a stronger signal than a
$\MinusSign$ edge. Similarly, we have $\bfw^- > \bfw^+$ when $q^- <
\nicefrac{1}{2} < p^+$ (case (ii) above);  in this case, a $\MinusSign$ edge
gives a stronger signal than a $\PlusSign$ edge.



\section{Algorithm}

In this section, we present an approximation algorithm for Correlation
Clustering with Asymmetric Classification Errors.
The algorithm first solves a standard LP relaxation and assigns every edge a length of $x_{uv}$ (see Section~\ref{sec:LP}).
Then, one by one it creates new clusters and removes them from the graph. The algorithm creates a cluster $C$ as follows. It picks a random vertex $p$, called a pivot, among yet unassigned vertices and a random number $R\in [0,1]$. Then, it adds the pivot $p$ and all vertices $u$ with $f(x_{pu}) \leq R$ to $C$, where $f:[0,1]\to [0,1]$ is a properly chosen function, which we define below.
We give a pseudo-code for this algorithm in Algorithm~\ref{alg:ApprAlg}.

\begin{algorithm}[tb]
   \caption{Approximation Algorithm}
   \label{alg:ApprAlg}
\begin{algorithmic}
   \INPUT An instance of Correlation Clustering with Asymmetric Weights $G=(V,E^+,E^-, \bfw_e)$.
   \STATE Initialize $t=0$ and $V_t=V$.
   \WHILE{$V_t \neq \varnothing$}
   \STATE Pick a random pivot $p_t\in V_t$.
   \STATE Choose a radius $R$ uniformly at random in $[0,1]$.
   \STATE Create a new cluster $S_t$; add the pivot $p_t$ to $S_t$.
   \FORALL{$u\in V_t$}
   \IF{$f(x_{p_{t}u})\leq R$}
   \STATE Add $u$ to $S_t$.
   \ENDIF
   \ENDFOR
   \STATE Let $V_{t+1} = V_t\setminus S_t$ and $t = t+1$.
   \ENDWHILE
   \OUTPUT clustering $\calS=(S_0,\dots, S_{t-1})$.
\end{algorithmic}
\end{algorithm}

Our algorithm resembles the LP-based correlation clustering algorithms by
\citet{ACN08} and \citet{CMSY15}. However, a crucial difference between our
algorithm and above mentioned algorithms is that our algorithm uses a
``dependant'' rounding. That is, if for two edges $pv_1$ and $pv_2$, we have
$f(x_{pv_1})\leq R$ and $f(x_{pv_2})\leq R$ at some step $t$ of the algorithm
then both $v_1$ and $v_2$ are added to the new cluster $S_t$. The algorithms
by \citet{ACN08} and \citet{CMSY15} make decisions on whether to add $v_1$ to
$S_t$ and $v_2$ to $S_t$, independently. Also, the choice of the function $f$
is quite different from the functions used by~\citet{CMSY15}. In fact, it is
influenced by the paper by~\citet*{GVY96}.
\subsection{Linear Programming Relaxation}\label{sec:LP}
In this section, we describe a standard linear programming (LP) relaxation for Correlation
Clustering which was introduced by~\citet*{CGW03}. We first give an integer programming formulation
of the Correlation Clustering problem. For every pair of vertices $u$ and $v$, the integer program (IP)
has a variable $x_{uv}\in \{0,1\}$, which indicates whether $u$ and $v$ belong to the same cluster:
\begin{itemize}
  \item $x_{uv}=0$, if $u$ and $v$ belong to the same cluster; and
  \item $x_{uv}=1$, otherwise.
\end{itemize}
We require that $x_{uv}=x_{vu}$, $x_{uu}=0$
and all $x_{uv}$ satisfy the triangle inequality. That is, $x_{uv} + x_{vw}\geq x_{uw}$.

Every feasible IP solution $x$ defines a partitioning $\calS=(S_1,\dots,S_T)$ in which
two vertices $u$ and $v$ belong to the same cluster if and only if $x_{uv} = 0$.
A positive edge $uv$ is in disagreement with this partitioning if and only if $x_{uv} = 1$;
a negative edge $uv$ is in disagreement with this partitioning if and only if $x_{uv} = 0$.
Thus, the cost of the partitioning is given by the following linear function:
$$\sum_{uv\in E^+} \bfw_{uv} x_{uv} + \sum_{uv\in E^-} \bfw_{uv} (1 - x_{uv}).$$

We now replace all integrality constraints $x_{uv}\in \{0,1\}$ in the integer program
with linear constraints $x_{uv}\in [0,1]$ . The obtained linear program is given in Figure~\ref{fig:LP}.
In the paper, we refer to each variable $x_{uv}$ as the length of the edge $uv$.

\begin{figure}
\openLP

$$\min \sum_{uv\in E^+} \bfw_{uv} x_{uv} + \sum_{uv\in E^-} \bfw_{uv} (1 - x_{uv}).$$

\noindent\textbf{subject to}
\begin{align*}
x_{uw}&\leq x_{uv}+x_{vw}&\text{for all } u,v,w\in V\\
x_{uv}&=x_{vu}&\text{for all } u,v\in V\\
x_{uu}&=0&\text{for all } u\in V\\
x_{uv}&\in [0,1]&\text{for all } u,v\in V
\end{align*}
\closeLP
\caption{LP relaxation}\label{fig:LP}
\end{figure}

\section{Analysis of the Algorithm}
The analysis of our algorithm follows the general approach proposed by~\citet*{ACN08}.
\citet{ACN08} observed that in order to get upper bounds on the approximation
factors of their algorithms, it is sufficient to consider how these algorithms
behave on triplets of vertices. Below, we present their method adapted to our
settings. Then, we will use Theorem~\ref{thm:L4} to analyze our algorithm.

\subsection{General Approach: Triple-Based Analysis}
\label{triple_Analysis}
Consider an instance of Correlation Clustering  $G=(V,E^+,E^-)$ on three vertices $u$, $v$, $w$. Suppose that the edges $uv$, $vw$, and $uw$ have
signs $\sigma_{uv}, \sigma_{vw}, \sigma_{uw}\in \{\pm\}$, respectively. We shall call this instance a triangle $(u,v,w)$ and refer to
the vector of signs $\sigma =(\sigma_{vw}, \sigma_{uw}, \sigma_{uv})$ as the signature of the triangle~$(u,v,w)$.

Let us now assign arbitrary lengths $x_{uv}$, $x_{vw}$, and $x_{uw}$ satisfying the triangle inequality
to the edges $uv$, $vw$, and $uw$ and run one iteration of our algorithm on the triangle $uvw$
(see Algorithm~\ref{alg:alg-one-step}).

\begin{algorithm}[tb]
   \caption{One iteration of Algorithm~\ref{alg:ApprAlg} on triangle $uvw$}
   \label{alg:alg-one-step}
\begin{algorithmic}
   \STATE Pick a random pivot $p\in \{u,v,w\}$.
   \STATE Choose a random radius $R$ with the uniform distribution in $[0,1]$.
   \STATE Create a new cluster $S$. Insert $p$ in $S$.
   \FORALL{$a \in \{u,v,w\}\setminus\{p\}$}
   \IF{$f_{\alpha}(x_{pa})\leq R$}
   \STATE Add $a$ to $S$ .
   \ENDIF
   \ENDFOR
\end{algorithmic}
\end{algorithm}

We say that a positive edge $uv$ is in disagreement with $S$ if $u\in S$ and $v\notin S$ or $u\notin S$ and $v\in S$. Similarly,
a negative edge $uv$ is in disagreement with $S$ if $u,v\in S$. Let $cost(u,v\given w)$ be the probability that the
edge $(u,v)$ is in disagreement with $S$ given that $w$ is the pivot.
$$cost(u,v\given w) =
\begin{cases}
\pr(u\in S, v\notin S \text{ or } u\notin S, v\in S\given p = w),& \text{if } \sigma_{uv} = \PlusSign;\\
\pr(u\in S, v\in S\given p = w),& \text{if } \sigma_{uv} = \MinusSign.
\end{cases}$$

Let $lp(u,v\given w)$ be the LP contribution of the edge $(u,v)$ times the probability of it being removed, conditioned on $w$ being the pivot.
$$lp(u,v\given w) =
\begin{cases}
x_{uv}\cdot \pr(u\in S \text{ or } v\in S \given p = w),& \text{if } \sigma_{uv} = \PlusSign;\\
(1-x_{uv}) \cdot \pr(u\in S \text{ or } v\in S\given p = w),& \text{if } \sigma_{uv} = \MinusSign.
\end{cases}$$

We now define two functions $ALG^{\sigma}(x,y,z)$ and $LP^{\sigma}(x,y,z)$. To this end,
construct a triangle $(u,v,w)$ with signature $\sigma$ edge lengths $x,y,z$ (where
$x_{vw} = x$, $x_{uw} = y$, $x_{uv} = z$).
Then,
\begin{align*}
ALG^{\sigma}(x,y,z) &= \bfw_{uv}\cdot cost(u,v\given w) + \bfw_{uw}\cdot cost(u,w\given v) + \bfw_{vw}\cdot cost(v,w\given u);\\
LP^{\sigma}(x,y,z) &= \bfw_{uv}\cdot lp(u,v\given w) + \bfw_{uw}\cdot lp(u,w\given v) + \bfw_{vw}\cdot lp(v,w\given u).
\end{align*}

We will use the following theorem from the paper by \citet*{CMSY15} (Lemma~4) to analyze our algorithm.
This theorem was first proved by~\citet*{ACN08} but it was not stated in this form in their paper.
\begin{theorem}[see \cite{ACN08} and \cite{CMSY15}]\label{thm:L4}
Consider a function $f_{\alpha}$ with $f_{\alpha}(0) = 0$. If for all signatures $\sigma=(\sigma_1,\sigma_2,\sigma_3)$
(where each $\sigma_i\in \{\pm\}$) and edge lengths $x$, $y$, and $z$ satisfying the triangle inequality,
we have $ALG^{\sigma}(x,y,z)\leq \rho LP^{\sigma}(x,y,z)$, then the approximation factor of the algorithm is at most $\rho$.
\end{theorem}
\subsection{Analysis of the Approximation Algorithm}
\begin{proof}[Proof of Theorem~\ref{thm:main}] Without loss of generality we
assume that the scaling parameter $\bfw$ is $1$. We use different functions for 
$\alpha \leq 0.169$ and $\alpha \geq 0.169$. Let $A = 3 + 2\log_e 1/\alpha$. 
For $\alpha \leq 0.169$, we define
$f_{\alpha}(x)$ as follows (see Figure~\ref{fig:plot-f}):
$$
    f_{\alpha}(x)= \left\{
\begin{array}{ll}
      1-e^{-Ax}, & \text{if }0\leq x<\frac{1}2-\frac{1}{2A}; \\
      1, & \text{otherwise};
\end{array}
\right.
$$
and, for $\alpha \geq 0.169$, we define $f_{\alpha}(x)$ as follows:
\begin{align*}
    f_{\alpha}(x)= \left\{
\begin{array}{ll}
      0, &\mbox{ if } x<\frac{1}{A} \\
      \frac{1-\alpha}{3}, &\mbox{ if } \frac{1}{A}\leq x<\frac{1}{2}-\frac{1}{2A} \\
      1, &\mbox{ if } x\geq \frac{1}{2}-\frac{1}{2A} \\
\end{array}
\right.
\end{align*}

\begin{figure}[t]
\centering
\begin{tikzpicture}
\begin{axis}[
   xlabel=$x$,
   ylabel=$f_{\alpha}(x)$,
   grid=both,
   minor grid style={gray!25},
   major grid style={gray!25},
   width=\linewidth,
    legend pos=south east,
    xtick distance=0.25,
    minor tick num=1,
    legend cell align={left}
]
\addplot[domain=0:0.4702,line width=0.5pt,solid,color=black] {1 - exp(-16.815*x) };
\addlegendentry{\tiny $f_{\alpha}$ for $\alpha = 0.001$};
\addplot[line width=0.5pt,dashed,color=black] %
   table[x=x,y=f,col sep=comma]{data/optimal-f-plot-0.001.csv};
\addlegendentry{\tiny $f_{opt}$ for $\alpha = 0.001$};
\addplot[domain=0:0.459,line width=0.5pt,solid,color=red] {1 - exp(-12.21*x) };
\addlegendentry{\tiny $f_{\alpha}$ for $\alpha = 0.01$};
\addplot[line width=0.5pt,dashed,color=red] %
   table[x=x,y=f,col sep=comma]{data/optimal-f-plot-0.01.csv};
\addlegendentry{\tiny $f_{opt}$ for $\alpha = 0.1$};
\addplot[domain=0:0.434,line width=0.5pt,solid,color=blue] {1 - exp(-7.605*x) };
\addlegendentry{\tiny $f$ for $\alpha = 0.1$};
\addplot[line width=0.5pt,dashed,color=blue] %
   table[x=x,y=f,col sep=comma]{data/optimal-f-plot-0.1.csv};
\addlegendentry{\tiny $f_{opt}$ for $\alpha = 0.1$};

\addplot[domain=0.4702:1,line width=0.6pt,color=black,dotted,forget plot] {1};
\addplot[domain=0.459:1,line width=0.6pt,color=orange,dotted,forget plot] {1};
\addplot[domain=0.434:1,line width=0.6pt,color=red,dotted,forget plot] {1};
\end{axis}
\node [color=black, rotate=85] at (0.55,2) {\tiny $0.001$};
\node [color=red, rotate=57] at (1.344,4.46) {\tiny $0.01$};
\node [color=blue, rotate=73] at (1.087, 3.23) {\tiny $0.1$};
\node [color=black, rotate=82] at (0.75,1) {\tiny $0.001$};
\node [color=red, rotate=67] at (1.7,3.46) {\tiny $0.01$};
\node [color=blue, rotate=75] at (1.55, 2.2) {\tiny $0.1$};
\end{tikzpicture}
\caption{This plot shows functions $f_{\alpha}(x)$ used in the proof of Theorem~\ref{thm:main} for $\alpha \in\{0.001, 0.01, 0.1\}$. Additionally, it shows optimal functions $f_{opt}(x)$
(see Section~\ref{sec:optimal} for details).
Note that every function $f_{\alpha}(x)$, including $f_{opt}(x)$, has a discontinuity at point $\tau = \nicefrac{1}{2} - \nicefrac{1}{2A}$; for $x \geq \tau$, $f_{\alpha}(x) = 1$.}
\label{fig:plot-f}
\end{figure}
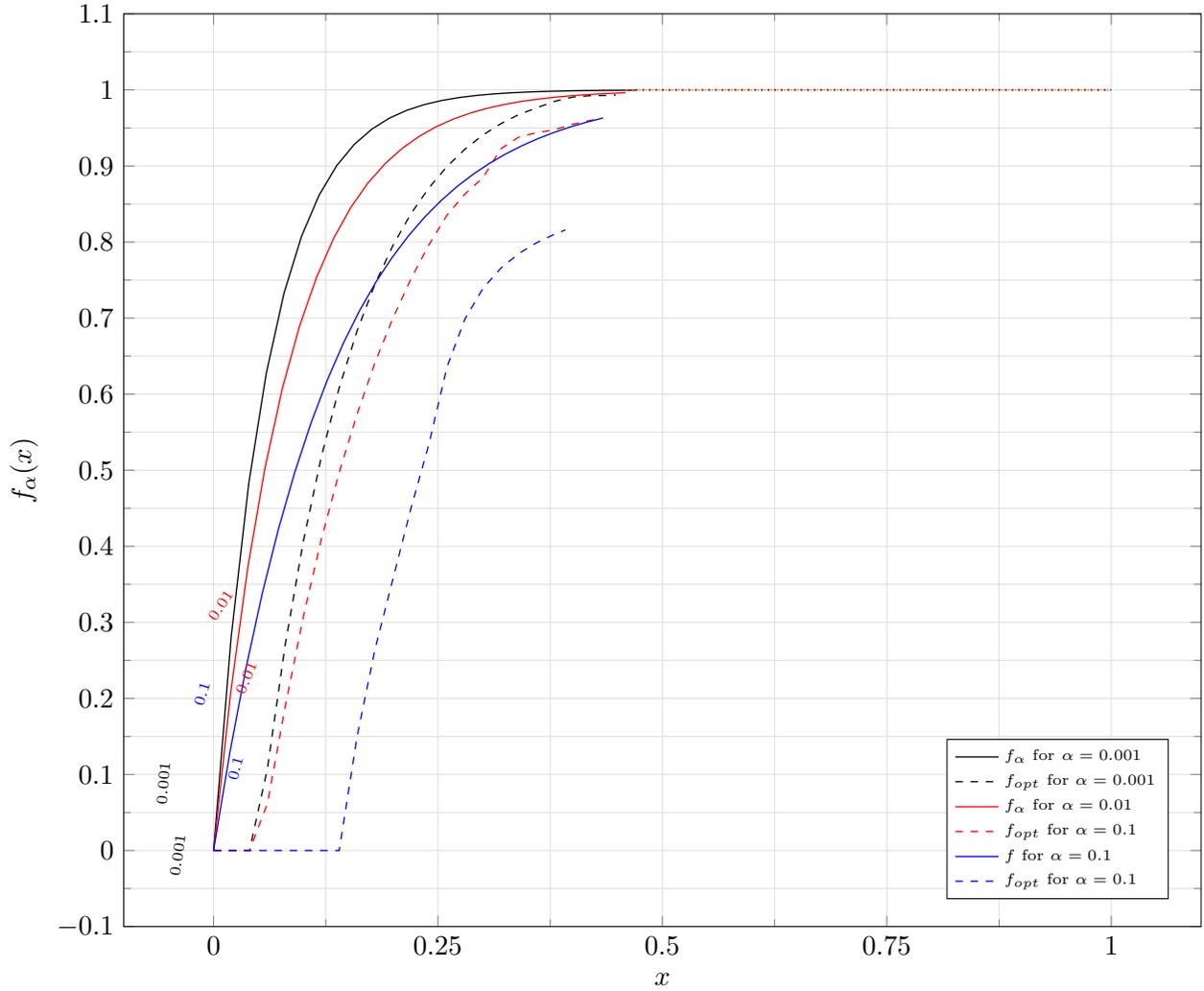

Our analysis of the algorithm relies on Theorem~\ref{thm:L4}. We will show that for every triangle
$(u_1,u_2,u_3)$ with edge lengths $(x_1,x_2,x_3)$ (satisfying the triangle inequality) and signature $\sigma = (\sigma_1,\sigma_2, \sigma_3)$,
we have
\begin{equation}\label{eq:analysis-main}
ALG^{\sigma}(x_1,x_2,x_3)\leq A\cdot LP^{\sigma}(x_1,x_2,x_3).
\end{equation}
Therefore, by Theorem~\ref{thm:L4}, our algorithm gives an $A$-approximation.

Without loss of generality, we assume that $x_1\leq x_2\leq x_3$. When $i\in\{1,2,3\}$ is fixed, we will denote the other two elements of $\{1,2,3\}$ by $k$ and $j$, so that $j < k$.
For $i\in\{1,2,3\}$, let $e_i = (u_j,u_k)$ (the edge opposite to $u_i$), $w_i = \bfw_{e_i}$, $x_i = x_{u_ju_k}$, $y_i = f_{\alpha}(x_i)$, and 
$$t_i = A\cdot lp(u_j, u_k | u_i)-cost(u_j, u_k | u_i).$$ 

Observe that (\ref{eq:analysis-main}) is equivalent to 
the inequality $w_1t_1 + w_2 t_2 + w_3 t_3 \geq 0$. We now prove that this inequality 
always holds.
\begin{lemma}\label{lem:wr-geq-0} We have
\begin{equation}\label{eq:analysis-main-alt}
w_1t_1 + w_2 t_2 + w_3 t_3 \geq 0
\end{equation}
\end{lemma}

We express each $t_i$ in terms of $x_i$'s and $y_i$'s.
\begin{claim}\label{cl:what_t_i_looks_like} For every $i\in\{1,2,3\}$, we have
$$
    t_i =
\begin{cases}
    A(1 - y_j) x_i-(y_k - y_j), &  \text{if } \sigma_i = \PlusSign\\
    A(1-y_j)(1-x_i)-(1 - y_k) , & \text{if } \sigma_i = \MinusSign
\end{cases}
$$
\end{claim}
\begin{proof}
If $\sigma_i = \PlusSign$, then
\begin{align*}
t_i &= A\cdot lp(u_j, u_k | u_i)-cost(u_j, u_k | u_i)\\
&=A x_{u_ju_k}\cdot \pr(u_j\in S \text{ or } u_k\in S \given p = u_i] -\pr(u_j\in S, u_k\notin S \text{ or } u_j\notin S, u_k\in S \given p = u_i]\\
&= A x_i \cdot \pr(f_{\alpha}(x_k) \leq R \text{ or } f_{\alpha}(x_j) \leq R ) - \pr(f_{\alpha}(x_k) \leq R < f_{\alpha}(x_j) \text{ or } f_{\alpha}(x_j) \leq R < f_{\alpha}(x_k) )\\
&= A x_i (1 - y_j) - (y_k - y_j),
\end{align*}
where we used that $y_k = f_{\alpha}(x_k) \geq f_{\alpha}(x_j) = y_j$ (since $x_k \geq x_j$ and $f_{\alpha}(x)$ is non-decreasing).

If $\sigma_i = \MinusSign$, then similarly to the previous case, we have
\begin{align*}
t_i &= A\cdot lp(u_j, u_k | u_i)-cost(u_j, u_k | u_i)\\
&= A (1-x_{u_ju_k}) \cdot \pr(u_j\in S \text{ or } u_k\in S \given p = u_i)- \pr(u_j\in S, u_k\in S\given p = w) \\
&= A (1-x_i) \cdot \pr(f_{\alpha}(x_k)\leq R \text{ or } f_{\alpha}(x_j) \leq R) -
 \pr(f_{\alpha}(x_k) \leq R, f_{\alpha}(x_j)\leq R) \\
&= A (1-x_i) \cdot (1 - y_j) - (1 - y_k).
\end{align*}
\end{proof}

We say that edge $e_i$ \textit{pays for itself} if $t_i \geq 0$. Note
that if all edges $e_1, e_2, e_3$ pay for themselves then the desired inequality
$(\ref{eq:analysis-main-alt})$ holds. First, we show that all negative edges pay
for themselves.

\begin{claim}\label{cl:negative_pays}
If $\sigma_i = \MinusSign$, then $t_i\geq 0$.
\end{claim}
\begin{proof}
By Claim~\ref{cl:what_t_i_looks_like}, $t_i = A(1-y_j)(1-x_i)- 1-y_k$.
Thus, we need to show that $A(1-y_j)(1-x_i)\geq 1-y_k$. 
If $x_k\geq \frac{1}{2} - \frac{1}{2A}$ then $y_k = 1$, and the inequality trivially holds. 
If $x_k< \frac12 - \frac{1}{2A}$, then using $x_j \leq x_k$, we get
$$A >\frac{1}{1-2x_k} \geq \frac{1}{1-x_k - x_j} \geq
\frac{1}{1-x_i},$$
here we used the triangle inequality $x_k + x_j \geq x_i$. Thus
$$
   A(1-y_j)(1-x_i)\geq A(1-y_k)(1-x_i)\geq 1-y_k.
$$
\end{proof}

We now show that for short edges $e_i$, it is sufficient to consider only the case 
when $\sigma_i = \PlusSign$. Specifically, we prove the following claim.
\begin{claim}\label{positive_is_enough}
Suppose that $x_i<\frac{1}{2}-\frac{1}{2A}$. If (\ref{eq:analysis-main-alt}) holds for $\sigma$ with $\sigma_i = \PlusSign$, then (\ref{eq:analysis-main-alt}) also holds for $\sigma'$ obtained from $\sigma$ by changing the sign of $\sigma_i$ to \MinusSign.
\end{claim}
\begin{proof}
To prove the claim, we show that the value of $t_i$ is greater for $\sigma'$ than for $\sigma$. That is, 
$$A(1-y_j)x_i -(y_k-y_j)< A(1-y_j)(1-x_i)-(1-y_k).$$
Note that the values of $t_j$ and $t_k$ do not depend on $\sigma_i$ and 
thus do not change if we replace $\sigma$ with $\sigma'$.
Since $f_{\alpha}$ is non-decreasing and $x_j\leq x_k$, we have $y_j\leq y_k$.
Hence,
$$
x_i< \frac{1}{2}-\frac{1}{2A}=\frac{1}{2}+\frac{1}{2A}-\frac{1}{A}\leq \frac{1}{2}+\frac{1}{2A}-\frac{(1-y_k)}{A(1-y_j)}.
$$
Thus, 
$$ 2A(1-y_j)x_i<A(1-y_j)+1-y_j-2(1-y_k).$$
Therefore, $$A(1-y_j)x_i -(y_k-y_j)< A(1-y_j)(1-x_i)-(1-y_k),$$ as required.
\end{proof}

Unlike negative edges, positive edges do not necessarily pay for themselves. 
We now prove that positive edges of length at least $1/A$ pay for themselves.
\begin{claim}\label{cl:positive_pays}
If $\sigma_i = \PlusSign$ and $x_i \geq 1/A$, then $t_i\geq 0$.
\end{claim}
\begin{proof}
We have,
$$t_i = A(1 - y_j) x_i-(y_k - y_j) \geq (1 - y_j)-(y_k - y_j) = 1 - y_k \geq 0.$$
\end{proof}

We now separately consider two cases $\alpha \leq 0.169$ and $\alpha \geq 0.169$.
\let\qed\relax 
\end{proof}

\subsection{Analysis of the Approximation Algorithm for \texorpdfstring{$\alpha \leq 0.169$}{α ≤ 0.169}}
First, we consider the case of $\alpha \leq 0.169$. 
\begin{proof}[Proof of Lemma~\ref{lem:wr-geq-0} for $\alpha \leq 0.169$]
We first show that if $x_3<\frac{1}{2}-\frac{1}{2A}$, then all three edges $e_1$, $e_2$, and $e_3$ 
pay for themselves.
\begin{claim}\label{cl:all_pay} If $x_3<\frac{1}{2}-\frac{1}{2A}$, then 
$t_i \geq 0$ for every $i$.
\end{claim}
\begin{proof}
Since $x_3 < \frac{1}{2}-\frac{1}{2A}$, for every $i\in\{1,2,3\}$ we have $x_i
< \frac{1}{2}-\frac{1}{2A}$ and thus $y_i \equiv f_{\alpha}(x_i) = 1- e^{-Ax_i}$.
We show that $t_i \geq 0$ for all $i$. Fix $i$. If $\sigma_i = \MinusSign$,
then, by Claim~\ref{cl:negative_pays}, $t_i \geq 0$. If $\sigma_i = \PlusSign$, then
\begin{multline*}
y_k-y_j= e^{-Ax_j}-e^{-Ax_k}=e^{-Ax_j}\left(1-e^{-A\left(x_k-x_j\right)}\right)
\leq \\
\leq 
e^{-Ax_j}A(x_k-x_j)\leq e^{-Ax_j}Ax_i =A(1-y_j)x_i,
\end{multline*}
where the first inequality follows from the inequality $1-e^{-x}\leq x$, and
the second inequality follows from the triangle inequality. Thus, $t_i =
A(1 - y_j) x_i-(y_k - y_j) \geq 0$.
\end{proof}
We conclude that if $x_3< \frac{1}{2} -\frac{1}{2A}$, then
(\ref{eq:analysis-main-alt}) holds. The case $x_3< \frac{1}{2} -\frac{1}{2A}$ is the most interesting 
case in the analysis; the rest of the proof is more technical. As a side note, 
let us point out that Theorem~\ref{thm:main} has dependence $A = 3 + 2\log_e 1/\alpha$ 
because (i) $f_{\alpha}(x)$ must be equal to $C - e^{-Ax}$ or a slower growing function so that
Claim \ref{cl:all_pay} holds (ii) Theorem~\ref{thm:L4} requires that $f_{\alpha}(0) = 0$,
and finally (iii) we will need below that
$1-f\left(\frac{1}{2}-\frac{3}{2A}\right) \leq \alpha$.

\medskip

From now on, we assume that $x_3\geq  \frac12 -\frac{1}{2A}$ and, consequently,
$y_3 = f(x_3)=1$. Observe that if $x_1\geq \frac{1}{A}$, then all $x_i \geq \frac{1}{A}$ and
thus, by Claims~\ref{cl:negative_pays} and~\ref{cl:positive_pays}, all $t_i \geq 0$
and we are done. Similarly, if $x_2\geq \frac{1}{2}-\frac{1}{2A}$, then 
$x_2 \geq \frac{1}{A}$ (since $A \geq 3$). Hence,  $t_2\geq 0$ and $t_3 \geq 0$;
additionally, $y_2=y_3=1$. Thus $t_1 = 0$ and inequality~(\ref{eq:analysis-main-alt}) holds. 
Therefore, it remains to show that inequality~(\ref{eq:analysis-main-alt}) holds
when
$$x_1 < \frac1A,\quad  x_2 < \frac{1}{2}-\frac{1}{2A},
\text{ and } x_3 \geq \frac{1}{2}-\frac{1}{2A}.$$
By Claim \ref{positive_is_enough}, we
may also assume that $\sigma_1 = \PlusSign$ and $\sigma_2 = \PlusSign$.
Since $\alpha \leq 0.169$, we have $A > 5$ and
$$x_2 \geq x_3 - x_1 \geq \bigg(\frac12 - \frac{1}{2A}\bigg) - \frac1A >
\frac1A\text{ and }x_3 \geq \frac{1}{2}-\frac{1}{2A} > \frac1A.$$
Thus, by Claims~\ref{cl:negative_pays} and~\ref{cl:positive_pays}, $t_2 \geq 0$ and $t_3
\geq 0$. Hence, $w_2 t_2 + w_3 t_3 \geq \alpha(w_2+w_3)$.
Also, recall that $e_1$ is a positive edge and thus $w_1 \leq 1$.
Therefore, it is sufficient to show that
\begin{equation}\label{eq:analysis-main-alt-reweighted}
t_1 \geq -\alpha(t_2 + t_3).
\end{equation}

Now we separately consider two possible signatures $\sigma = (\PlusSign,\PlusSign,\PlusSign)$
and $\sigma = (\PlusSign,\PlusSign,\MinusSign)$.

\medskip

\noindent\textbf{First, assume that $\sigma =
(\PlusSign,\PlusSign,\PlusSign)$.} We need to show that
$$
A(1 - y_2) x_1 - (1-y_2)\geq \alpha \bigg( 
(1 - y_1) + (y_2 - y_1)  - A(1 - y_1) x_2 - A(1 - y_1) x_3
\bigg).
$$
Here, we used that $y_3=1$.
Note that $x_2\geq x_3-x_1\geq \frac{1}{2}-\frac{1}{2A}-\frac{1}{A}=\frac{1}{2}-\frac{3}{2A}$. Therefore,
\begin{align*}
1-y_2\leq 1-\left(1-e^{-A\left(\frac{1}2-\frac{3}{2A}\right)}\right)&=e^{-\frac{3}{2}-\log_e\frac{1}{\alpha}+\frac{3}{2}}=e^{-\log_e\frac{1}{\alpha}}=\alpha.
\end{align*}
Thus, $(1-y_2)+\alpha(1 - y_1) + \alpha(y_2 - y_1)\leq \alpha y_2 +2\alpha(1-y_1)$. To finish the analysis of the case $\sigma = (\PlusSign,\PlusSign,\PlusSign)$, it
is sufficient to show that
\begin{align*}
    \alpha y_2 +2\alpha(1-y_1)\leq& A(1 - y_2) x_1+\alpha A(1 - y_1) x_2+ \alpha A(1 - y_1) x_3.
\end{align*}
This inequality immediately follows from the following claim (we simply need to add up (\ref{refined_x_2}) and (\ref{refined_x_3}) and multiply the result by $\alpha$).
\begin{claim}
For $c =0.224$, we have
\begin{align}
    (2-c)(1-y_1) &\leq  A(1-y_1) x_2 \label{refined_x_2};\text{ and}\\
     y_2 + c(1-y_1)&\leq A(1-y_1) x_3 \label{refined_x_3}.
\end{align}
\end{claim}
\begin{proof}
Since $c\geq 2-\log_e\frac{1}{0.169} \geq 2-\log_e\frac{1}{\alpha}$ (recall that $\alpha \leq 0.169$), we have
$$2-c\leq \log_e\frac{1}{\alpha}= \frac{A}{2}-\frac{3}{2}\leq Ax_2.$$
Therefore, (\ref{refined_x_2}) holds.
We also have, 
$$c\leq 0.169 + \log_e \frac{1}{0.169}+1-e \leq \alpha+\log_e\frac{1}{\alpha}+1-e.$$
Thus,
$e-\alpha\leq\frac{A}{2}-\frac{1}{2}-c \leq Ax_3 - c$.
Therefore,
\begin{align}\label{intermediate}
    e^{-1}\left(Ax_3-c\right)&\geq 1-\alpha e^{-1}= 1-e^{-A\left(\frac{1}2-\frac{1}{2A}\right)}\geq y_2,
\end{align}
where we used that $x_2< \frac{1}{2}-\frac{1}{2A}$ and $y_2 = f_{\alpha}(x_2)  = 1 - e^{-Ax_2}$. 
Observe that from inequalities (\ref{intermediate}) and $x_1< \frac{1}{A}$ it follows that
$$y_2\leq \left(1-f\Big(\frac{1}{A}\Big)\right)(Ax_3-c)\leq (1-y_1)(Ax_3-c),$$
which implies (\ref{refined_x_3}).
\end{proof}

\medskip

\noindent\textbf{Now, assume that $\sigma = (\PlusSign,\PlusSign,\MinusSign)$.} We need to prove the following inequality,
\begin{equation}\label{e_3_negative}
(1-y_2)+\alpha (1 - y_1+1-y_2)\leq A(1 - y_2) x_1+\alpha A(1 - y_1) (x_2+1-x_3).
\end{equation}
As before,
\begin{equation}\label{e_3_negative_1}
(1-y_2)+\alpha (1 - y_1+1-y_2)\leq \alpha + \alpha (1 - y_1+1-y_2)\leq \alpha +2\alpha(1 - y_1).
\end{equation}
On the other hand,
\begin{align}\label{e_3_negative_2}
    A(1 -  y_2) x_1+\alpha A(1 - y_1)(x_2+1-x_3)&\geq \alpha A (1 - y_1) (1-x_1+x_1+x_2-x_3)\nonumber\\
    &\geq \alpha A (1 - y_1) (1-x_1)\nonumber\\
    &\geq \alpha A (1 - y_1) \left(1-\frac{1}{A}\right)\nonumber\\
    &=\alpha (1 - y_1)(A-1)
\end{align}
where the second inequality is due to the triangle inequality, and the third inequality is due to $x_1<\frac{1}{A}$.
Finally, observe that $1\leq 2e^{-1}\log_e\frac{1}{\alpha}=e^{-1}(A-3)\leq (1-y_1)(A-3)$.
We get,
\begin{equation}\label{e_3_negative_3}
    \alpha (1 - y_1)(A-1)\geq \alpha +2\alpha(1 - y_1).
\end{equation}
Combining~(\ref{e_3_negative_1}), (\ref{e_3_negative_2}), and (\ref{e_3_negative_3}), we get~(\ref{e_3_negative}).
This concludes the case analysis and the proof of Theorem~\ref{thm:main} for the regime $\alpha\leq 0.169$. 
\let\qed\relax 
\end{proof}

\subsection{Analysis of the Approximation Algorithm for \texorpdfstring{$\alpha\geq 0.169$}{α ≥ 0.169}}
We now consider the case when $\alpha \geq 0.169.$ Observe that for $\alpha \geq 0.169$
\begin{equation}\label{large_alpha_ineq1}
A=3+2\log_e(1/\alpha) \geq \frac{6\alpha +3-(1-\alpha)^2}{3\alpha}
\end{equation}
and
\begin{equation}\label{large_alpha_ineq2}
\frac{1-\alpha}{3}\leq \frac{2\alpha}{1+\alpha}
\end{equation}
\begin{proof}[Proof of Lemma~\ref{lem:wr-geq-0} for $\alpha \geq 0.169$]
Observe that if $x_1\geq \frac{1}{A}$, then all $x_i \geq 1/A$ and thus, by Claims~\ref{cl:negative_pays} 
and~\ref{cl:positive_pays}, all $t_i \geq 0$ and we are done. Moreover, if $x_3 < \frac{1}{A}$ then
all $x_i < 1/A$ implying $y_i=0$ and thus, $t_i\geq 0$ for $\sigma_i=``+"$. 
This combined with Claim~\ref{cl:negative_pays} imply all $t_i\geq 0$ and we are done. 
Similarly, if $x_2\geq \frac{1}{2}-\frac{1}{2A}$, then $x_2\geq 1/A$ (since $A \geq 3$). 
Hence, $t_2\geq 0$ and $t_3 \geq 0$; additionally, we have $y_2=y_3=1$. 
Thus, $t_1 = 0$ and we are done.

Therefore, we will assume below that
\begin{align*}
x_1 < \frac{1}{A},\;\;x_2 < \frac{1}2-\frac{1}{2A},\;\; x_3\geq \frac{1}{A}.
\end{align*}

Furthermore, by Claim \ref{positive_is_enough}, we may assume $\sigma_1 = \PlusSign$ and $\sigma_2 = \PlusSign$. We consider four cases: (i) $x_2\geq \nicefrac{1}{A},\;x_3\geq \nicefrac{1}{2}-\nicefrac{1}{(2A)}$, (ii) $x_2< \nicefrac{1}{A},\;x_3\geq \nicefrac{1}{2}-\nicefrac{1}{(2A)}$, (iii) $x_2\geq \nicefrac{1}{A},\;x_3< \nicefrac{1}{2}-\nicefrac{1}{(2A)}$, and (iv) $x_2< \nicefrac{1}{A},\;x_3< \nicefrac{1}{2}-\nicefrac{1}{(2A)}$.

\medskip

\noindent\textbf{Consider the case $x_2\geq \frac{1}{A},\;x_3\geq \frac{1}{2}-\frac{1}{2A}.$} Then $y_1=0,\;y_2=\nicefrac{(1-\alpha)}{3},\;y_3=1.$ By Claims~\ref{cl:negative_pays} and~\ref{cl:positive_pays}, $t_2,t_3\geq 0$, and $e_2,e_3$ pay for themselves. If $t_1\geq 0$, we are done. 
So we will assume below that $t_1<0$. Then,
\begin{equation}\label{weightAssumption1}
w_1t_1+w_2t_2+w_3t_3\geq 1\cdot t_1+\alpha t_2+\alpha t_3
\end{equation}
(recall that we assume that $e_1$ is a positive edge and thus $w_1\leq 1$).

Now we separately consider two possible signatures $\sigma = (\PlusSign,\PlusSign,\PlusSign)$ and
$\sigma = (\PlusSign,\PlusSign,\MinusSign)$.

\medskip

\noindent\textbf{First, assume that $\sigma = (\PlusSign,\PlusSign,\PlusSign)$.} Because of~(\ref{weightAssumption1}),
to prove~(\ref{eq:analysis-main-alt}) it is sufficient to show
\begin{equation}\label{+++}
    (1-y_2)+\alpha +\alpha y_2
    \leq A(1-y_2)x_1+\alpha Ax_2 +\alpha Ax_3
\end{equation}

From~(\ref{large_alpha_ineq1}) it follows that
$$
    1+\alpha\leq\frac{(1-\alpha)^2}{3}+\alpha (A-1)
$$

which implies~(\ref{helper1}) due to $x_3\geq \frac{1}{2}-\frac{1}{2A}$
\begin{equation}\label{helper1}
    1+\alpha\leq \frac{(1-\alpha)^2}{3}+2\alpha Ax_3
\end{equation}

Observe that from~(\ref{helper1}) together with triangle inequality and $y_2=\frac{1-\alpha}{3}\leq 1-\alpha$ it follows that
\begin{equation*}
 1+\alpha\leq (1-\alpha)y_2+A(1-y_2)x_1-\alpha Ax_1
 +\alpha Ax_1+\alpha Ax_2+\alpha Ax_3
\end{equation*}
which is equivalent to~(\ref{+++}).

\medskip

\noindent\textbf{Now, assume that $\sigma = (\PlusSign,\PlusSign,\MinusSign)$.} Because of~(\ref{weightAssumption1}), to prove~(\ref{eq:analysis-main-alt}) it is sufficient to show
\begin{equation}\label{++-}
    (1-y_2)+\alpha+\alpha (1-y_2)\leq A(1-y_2)x_1+\alpha Ax_2+\alpha A(1-x_3)
\end{equation}
From~(\ref{large_alpha_ineq1}) and $y_2=\frac{1-\alpha}{3}$ it follows that

\begin{equation*}
1+2\alpha\leq\frac{(1-\alpha)^2}{3}+\alpha A \leq y_2(1 + \alpha) + \alpha A
\end{equation*}
Since $y_2 \leq 1-\alpha$, 
\begin{equation*}
(1 + 2\alpha) \leq (1+\alpha)y_2+A(1-y_2)x_1-\alpha Ax_1 +\alpha A,
\end{equation*}
Hence, using the triangle inequality,
\begin{equation*}
    1+2\alpha\leq (1+\alpha)y_2+A(1-y_2)x_1-\alpha Ax_1+\alpha A
    +\alpha Ax_1+\alpha Ax_2-\alpha Ax_3.
\end{equation*}
which is equivalent to~(\ref{++-}).
\medskip

\noindent\textbf{Consider the case $x_2< \frac{1}{A},\;x_3\geq \frac{1}{2}-\frac{1}{2A}.$} Then $y_1=y_2=0,\;y_3=1.$ Observe that $t_3\geq 0$ and $t_1,t_2<0$. Then,
\begin{equation}\label{weightAssumption2}
w_1t_1+w_2t_2+w_3t_3 \geq 1\cdot t_1+1\cdot t_2+\alpha t_3.
\end{equation}
(recall that we assume that $e_1,e_2$ are positive edges and thus $w_1,w_2\leq 1$).
Furthermore, since $x_3\geq \frac{1}{2}-\frac{1}{2A}$ we have
\begin{equation}\label{helper3}
    Ax_3\geq A(1-x_3)-1.
\end{equation}

From~(\ref{helper3}), we get that if (\ref{eq:analysis-main-alt}) holds for $\sigma$ with $\sigma_3 = \MinusSign$, then (\ref{eq:analysis-main-alt}) also holds for $\sigma'$ obtained from $\sigma$ by changing the sign of $\sigma_3$ to $\PlusSign.$ Thus without loss of generality $\sigma_3=\MinusSign$ and we only need to consider $\sigma = (\PlusSign,\PlusSign,\MinusSign).$ Then, because of~(\ref{weightAssumption2}), to prove~(\ref{eq:analysis-main-alt}) it is sufficient to show
\begin{equation}\label{++-2}
    1+1+\alpha \leq Ax_1+Ax_2+\alpha A(1-x_3).
\end{equation}
From~(\ref{large_alpha_ineq1}) it follows that
$$
    A\geq\frac{5+\alpha}{\alpha+1}
$$
which is equivalent to
\begin{equation}\label{helper4}
2+\alpha\leq \alpha A +(1-\alpha)(\frac{A}{2}-\frac{1}{2}).
\end{equation}
Observe that from~(\ref{helper4}) together with triangle inequality and $x_3\geq \frac{1}{2}-\frac{1}{2A}$ it follows that
\begin{equation*}
  2+\alpha \leq \alpha A+(1-\alpha)Ax_3= Ax_3+\alpha A(1-x_3)
  \leq Ax_1+Ax_2+\alpha A(1-x_3).
\end{equation*}

\medskip

\noindent\textbf{Consider the case $x_2\geq \frac{1}{A},\;x_3< \frac{1}{2}-\frac{1}{2A}$.} 
Then $y_1=0,\;y_3=\nicefrac{(1-\alpha)}{3}$. By Claim~\ref{positive_is_enough} we only need
to consider $\sigma=(\PlusSign,\PlusSign,\PlusSign).$ Then
by Claim~\ref{cl:positive_pays}, $t_2,t_3\geq 0$. Thus, if $t_1\geq 0$ then 
$w_1t_1+w_2t_2+w_3t_3\geq 0$. Let us assume that $t_1<0$. 
Since $e_1$ is a positive edge, we have $w_1\leq 1$. Thus,
\begin{equation*}
w_1 t_1+w_2 t_2+w_3 t_3\geq 1\cdot t_1+\alpha t_2+\alpha t_3
\end{equation*}
We need to show that the right hand side in the above inequality is non-negative. Replace $t_1$, $t_2$, and $t_3$ with the expressions from Claim~\ref{cl:what_t_i_looks_like}. Now to obtain~(\ref{eq:analysis-main-alt}), it is sufficient to prove that
\begin{equation}\label{+++2}
    y_3-y_2+\alpha y_3+\alpha y_2
    \leq A(1-y_2)x_1+\alpha Ax_2 +\alpha Ax_3
\end{equation}
Observe that since $x_3\geq \frac{1}{A}$ we have
\begin{equation}\label{helper5}
    2\alpha\leq (1-\alpha)y_2+2\alpha Ax_3.
\end{equation}
Inequalities~(\ref{helper5}) and~(\ref{large_alpha_ineq2}) imply
\begin{equation}\label{helper6}
    (1+\alpha)y_3\leq (1-\alpha)y_2+2\alpha Ax_3.
\end{equation}
Observe that from~(\ref{helper6}) together with triangle inequality and $ y_2\leq 1-\alpha$ it follows that
\begin{equation*}
    (1+\alpha)y_3\leq (1-\alpha)y_2+A(1-y_2)x_1-\alpha Ax_1
    +\alpha Ax_1+\alpha Ax_2+\alpha Ax_3
\end{equation*}
which is equivalent to~(\ref{+++2}).

\medskip

\noindent\textbf{Consider the case $x_2< \frac{1}{A},\;x_3< \frac{1}{2}-\frac{1}{2A}.$} Then $y_1=y_2=0.$ By Claim~\ref{positive_is_enough} we only need to consider $\sigma=(\PlusSign,\PlusSign,\PlusSign).$ Then by Claim~\ref{cl:positive_pays}, $t_3\geq 0$.

If $x_1\geq \nicefrac{y_3}{A}$ then $t_1,t_2\geq 0$ and we are done. Thus we assume $x_1< \nicefrac{y_3}{A}$ which implies $t_1<0$. We consider two 
different regimes: (i) $x_2\geq \nicefrac{y_3}{A}$ and (ii) $x_2<\nicefrac{y_3}{A}.$

\medskip

\noindent\textbf{First, assume that $x_2\geq \nicefrac{y_3}{A}$} which implies $t_2\geq 0.$ Then,
\begin{equation}\label{weightAssumption4}
w_1t_1+w_2t_2+w_3t_3\geq 1\cdot t_1+\alpha t_2+\alpha t_3
\end{equation}
(recall that we assume that $e_1$ is a positive edge and thus $w_1\leq 1$).

Because of~(\ref{weightAssumption4}), to prove~(\ref{eq:analysis-main-alt}) it is sufficient to show
\begin{equation}\label{x_2_greater_than_y_3/A}
    y_3+\alpha y_3\leq Ax_1+\alpha Ax_2+\alpha Ax_3
\end{equation}

Observe that by~(\ref{large_alpha_ineq2}) and $y_3=\nicefrac{(1-\alpha)}{3}$ we have
\begin{equation*}
    (1+\alpha)y_3\leq 2\alpha\leq 2\alpha Ax_3\leq \alpha Ax_3 +\alpha Ax_1+\alpha Ax_2
    \leq Ax_1+\alpha Ax_2+\alpha Ax_3
\end{equation*}
where the second inequality follows from $x_3\geq \frac{1}{A}$ and the third inequality follows from triangle inequality.

\medskip

\noindent\textbf{Now, assume that $x_2<\nicefrac{y_3}{A}$} which implies $t_2< 0.$ Then,
\begin{equation}\label{weightAssumption5}
w_1t_1+w_2t_2+w_3t_3\geq 1\cdot t_1+1\cdot t_2+\alpha t_3
\end{equation}
(recall that we assume that $e_1,e_2$ are positive edges and thus $w_1,w_2\leq 1$).

Because of~(\ref{weightAssumption5}), to prove~(\ref{eq:analysis-main-alt}) it is sufficient to show
\begin{equation}\label{x_2_less_than_y_3/A}
    2y_3\leq Ax_1+Ax_2+\alpha Ax_3
\end{equation}
Observe that by~(\ref{large_alpha_ineq2}) and $x_3\geq \frac{1}{A}$
\begin{equation*}
    2y_3\leq \frac{4\alpha}{1+\alpha}\leq 1+\alpha\leq(1+\alpha)Ax_3\leq Ax_1+Ax_2+\alpha Ax_3
\end{equation*}
where the last inequality follows from triangle inequality.

This concludes the case analysis and the proof of Theorem~\ref{thm:main} for the regime $\alpha\geq 0.169$.
\end{proof}

\section{Better approximation for values of \texorpdfstring{$\alpha$}{\unichar{"25B}} appearing in practice}\label{sec:optimal}
We note that the choice of function $f(x)$ in Theorem~\ref{thm:main} is somewhat suboptimal.  
The best function $f_{opt}(x)$ for our analysis of Algorithm~\ref{alg:ApprAlg} can be computed using linear programming (with high precision). Using this function $f_{opt}$,
we can achieve an approximation factor $A_{opt}$ better than the approximation factor
$A_{thm} = 3 + 2\log_e 1/\alpha$ guaranteed by Theorem~\ref{thm:main} (for $\alpha \neq 1$).\footnote{It is also
possible to slightly modify Algorithm~\ref{alg:ApprAlg} so that it gets approximation $A_{opt}$
without explicitly computing $f$. We omit the details here.}
While asymptotically $A_{thm}/A_{opt} \to 1$ as $\alpha \to 0$,
$A_{opt}$ is noticeably better than $A_{thm}$ for many values of $\alpha$ that are likely
to appear in practice (say, for $\alpha \in (10^{-8}, 0.1)$).
We list approximation factors $A_{thm}$ and $A_{opt}$ for several values of $\alpha$ in
Table~\ref{fig:table-A}; we also plot the dependence of $A_{thm}$ and
$A_{opt}$ on $\alpha$ in Figure~\ref{fig:plot-A}.

\begin{table}[t]
\caption{Approximation factors $A_{thm}$ and $A_{opt}$ for different $\alpha$-s.}
\label{fig:table-A}
\vskip 0.15in
\begin{center}
\begin{small}
\begin{tabular}{rrrr}
  \toprule
  $\log_e \nicefrac{1}{\alpha}$ & $\nicefrac{1}{\alpha}$ & $A_{thm}$ & $A_{opt}$ \\
  \midrule
    0 & 1 & 3 & 3\\
    1.61  & 5 & 6.22 & 4.32 \\
    2.30 & 10 & 7.61 & 4.63\\
    3.91 & 50 & 10.82 & 6.07\\
    4.61 & 100 & 12.21 & 6.78\\
    6.21 & 500 & 15.43 &  8.69\\
    6.91 &1000 & 16.82 & 9.62\\
    8.52 & 5\,000 & 20.03 & 11.9\\
    10 & 22\,026.5 & 23 & 14.2\\
    15 & $3.3 \times 10^6$ & 33 &  22.6\\
    20 & $4.9\times 10^8$ & 43 & 31.3\\
  \bottomrule
\end{tabular}
\end{small}
\end{center}
\vskip -0.1in
\end{table}

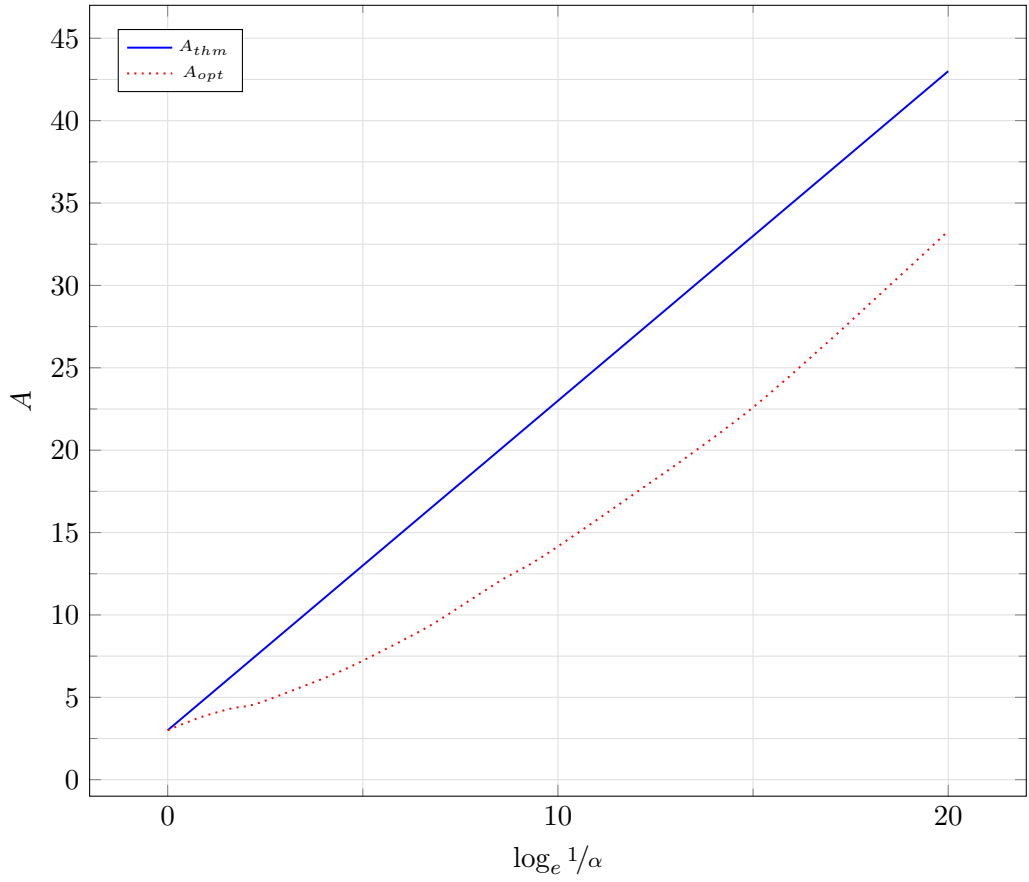
\begin{figure}[t]
\centering
\begin{tikzpicture}
\begin{axis}[
   xlabel=$\log_e \nicefrac{1}{\alpha}$,
   ylabel=$A$,
   grid=both,
   minor grid style={gray!25},
   major grid style={gray!25},
   width=0.85\linewidth,
    legend pos=north west,
    xtick distance=10,
    minor tick num=1
]
\addplot[domain=0:20,line width=0.75pt,solid,color=blue] {3 + 2*x};
\addlegendentry{\tiny $A_{thm}$};
\addplot[line width=0.75pt,dotted,color=red,smooth] %
   table[x=logalpha,y=A,col sep=comma]{data/optimal-A-plot-sparse.csv};
\addlegendentry{\tiny $A_{opt}$};
\end{axis}
\end{tikzpicture}
\caption{Plots of approximation factors $A_{thm}$ and $A_{opt}$.}
\label{fig:plot-A}
\end{figure}

\section{Analysis of the Algorithm for Complete Bipartite Graphs}
\begin{proof}[Proof of Theorem~\ref{thm:mainBipartite}]
The proof is similar to the proof of Theorem~\ref{thm:main}. Without loss of generality we
assume that the scaling parameter $\bfw$ is $1$. Define $f(x)$ as follows
$$
    f(x)= \left\{
\begin{array}{ll}
      1-e^{-Ax}, & \text{if }0\leq x<\frac{1}2-\frac{1}{2A} \\
      1, & \text{otherwise}
\end{array}
\right.
$$
where $A = 5 + 2\log_e 1/\alpha$. Our analysis of the algorithm relies on Theorem~\ref{thm:L4}. Since in the proof of Theorem~\ref{thm:L4}, we assumed that
all edges are present, let us add missing edges (edges inside parts) to the bipartite graph and assign them weight $0$; to be specific, we assume that they are positive edges. (It is important to note that Theorem~\ref{thm:L4} is true even when edges have zero weights). We will still refer to these edges as `missing edges'.

We will show that for every triangle $(u_1,u_2,u_3)$ with edge lengths $(x_1,x_2,x_3)$ (satisfying the triangle inequality) and signature $\sigma = (\sigma_1,\sigma_2, \sigma_3)$, we have
\begin{equation}\label{eq:analysis-main-Bipartite}
ALG^{\sigma}(x_1,x_2,x_3)\leq A\cdot LP^{\sigma}(x_1,x_2,x_3)
\end{equation}
Therefore, by Theorem~\ref{thm:L4}, our algorithm gives an $A$-approximation. In addition to Theorem~\ref{thm:L4} we use Claims~\ref{cl:what_t_i_looks_like}, \ref{cl:negative_pays}, \ref{positive_is_enough}, \ref{cl:positive_pays} and~\ref{cl:all_pay}. Recall that proofs of these claims rely on $f$ being non-decreasing which is satisfied by the above choice. Observe that~(\ref{eq:analysis-main-Bipartite}) is equivalent to
\begin{equation}\label{eq:analysis-main-alt-Bipartite}
\sum_{i=1}^3 w_i t_i \geq 0.
\end{equation}

Observe that if $x_1\geq \frac{1}{A}$, then all $x_i \geq \frac{1}{A}$ and
thus, by Claims~\ref{cl:negative_pays} and~\ref{cl:positive_pays}, all $t_i \geq 0$
and we are done. Similarly, if $x_2\geq \frac{1}2-\frac{1}{2A} \geq
\frac{1}{A}$ (since $A > 3$), then $t_2\geq 0$ and $t_3 \geq 0$;
additionally, $y_2=y_3=1$, thus $t_1 = 0$ and we are done. Furthermore, if $x_3<\frac{1}2-\frac{1}{2A}$ then all $x_i<\frac{1}2-\frac{1}{2A}$ and thus, by Claim~\ref{cl:all_pay}, all $t_i \geq 0$
and we are done. Therefore, we will assume below that $x_1 < \frac1A$, $x_2 < \frac{1}2-\frac{1}{2A}$, and $x_3 \geq \frac{1}2-\frac{1}{2A}$.
Further, by the triangle inequality $x_2\geq x_3 - x_1 \geq \frac{A-1}{2A} - x_1 \geq \frac{A-3}{2A}$. We have (here we use that $A \geq 5$),
$$ x_1 \leq \frac{1}{A} \leq \frac{A-3}{2A} \leq \frac{A-1}{2A} - x_1 \leq x_2 < \frac{A-1}{2A} \leq x_3 \leq x_1 + x_2.$$
We will use below that 
$$e^{A(x_2 - x_1)} \geq e^{A(\frac{A-1}{2A} - 2x_1)} = e^{2 + \log \frac{1}{\alpha} - 2Ax_1} = e^{2(1-Ax_1)}/\alpha \geq 1/\alpha.$$

By Claim \ref{positive_is_enough}, we may also assume that $\sigma_1 = \PlusSign$ and $\sigma_2 = \PlusSign$ (and since we assume that missing edges are positive). By Claims~\ref{cl:negative_pays} and~\ref{cl:positive_pays}, $t_2 \geq 0$ and $t_3\geq 0$ (edges $e_2$ and $e_3$ pay for themselves). If $t_1 \geq 0$, we are done. So we will assume below that $t_1 < 0$. Since $G$ is a complete bipartite graph, a triangle $(u_1,u_2,u_3)$ contains either (i) no edges or (ii) two edges. In case (i) we have $w_1=w_2=w_3=0$ and (\ref{eq:analysis-main-alt-Bipartite}) holds trivially. In case (ii) if $e_1$ is the missing edge then $w_1=0$ and since $t_2,t_3\geq 0$, (\ref{eq:analysis-main-alt-Bipartite}) holds trivially.
It remains to consider three signatures $\sigma = (\PlusSign,\PlusSign, \NoSign)$, $\sigma = (\PlusSign,\NoSign,\PlusSign)$
and $\sigma = (\PlusSign,\NoSign,\MinusSign)$ where $\NoSign$ denotes a missing edge (which by our assumption above is a positive edge).

\paragraph{First, assume that $\sigma =(\PlusSign,\PlusSign, \NoSign)$.}
By Claim~\ref{cl:what_t_i_looks_like}, $t_1 = A(1-y_2)x_1 - (1 - y_2) = - e^{-Ax_2}(1 - Ax_1)$ and $t_2 = A(1-y_1)x_2 - (1 - y_2) = e^{-Ax_1}(Ax_2 - 1)$.
Since $e_3$ is missing, $w_3 = 0$. We have, $w_1 t_1 + w_2 t_2 + w_3 t_3  \geq t_1 + \alpha t_2$ (here we used that $t_1 \leq 0$ and $t_2 \geq 0$). So it suffices to prove that $t_1 + \alpha t_2 > 0$ or, equivalently, $e^{Ax_2}(\alpha t_2 + t_1) \geq 0$. Using that $e^{A(x_2 - x_1)} \geq 1/\alpha$ and $x_2 \geq \frac{A-1}{2A} - x_1$, we get
$$
e^{Ax_2}(\alpha t_2 + t_1)  = \alpha e^{A(x_2 - x_1)}(Ax_2 - 1) - (1 - Ax_1) \geq \alpha \cdot \frac{1}{\alpha} \cdot \Bigl(A \bigl(\frac{A-1}{2A} - x_1\bigr) - 1\Bigr) + Ax_1 - 1
= \frac{A-5}{2} > 0,$$
as required.

\paragraph{Now, assume that $\sigma =(\PlusSign,\NoSign,\PlusSign)$.}
Now we have $t_1 = - e^{-Ax_2}(1 - Ax_1)$ (as before) and
$$t_3 = A(1 - y_1) x_3 - (y_2 - y_1) = A e^{-Ax_1} x_3 - (e^{-Ax_1} - e^{-Ax_2})= e^{-Ax_1} (Ax_3 -1) + e^{-Ax_2}.$$
We prove that $t_1 + \alpha t_3 \geq 0$ or, equivalently, $e^{Ax_2}(\alpha t_3 + t_1) \geq 0$.
Using that $e^{A(x_2 - x_1)} \geq 1/\alpha$ and $x_3 \geq \frac{A-1}{2A}$, we get
\begin{align*}
e^{Ax_2}(\alpha t_3 + t_1) &= \alpha \bigl(e^{A(x_2 - x_1)} (Ax_3 -1)  + 1\bigr) - (1 - Ax_1) \\
&\geq (Ax_3 -1)  + \alpha - (1 - Ax_1) > Ax_3  - 2 \geq \frac{A-1}{2} - 2 \geq 0,
\end{align*}
as required.

\paragraph{Finally, assume that $\sigma =(\PlusSign,\NoSign,\MinusSign)$.}
Now we have $t_1 = - e^{-Ax_2}(1 - Ax_1)$ (as before) and $t_3 = A(1 - y_1) (1 - x_3) - (1 - y_2) = A e^{-Ax_1} (1 - x_3) - e^{-Ax_2}$.
As in the previous case, we prove that $e^{Ax_2}(\alpha t_3 + t_1) \geq 0$. We have,
$$e^{Ax_2}(\alpha t_3 + t_1) = \alpha\bigl(A e^{A(x_2-x_1)} (1 - x_3) - 1\bigr) - (1 - Ax_1)\geq
\underbrace{\alpha\bigl(A e^{A(x_2-x_1)} (1 - x_1 - x_2) - 1\bigr) - (1 - Ax_1)}_{F(x_1,x_2)}.$$
Denote the expression on the right by $F(x_1, x_2)$. We now show that for a fixed $x_1$, $F(x_1, x_2)$ is an increasing function of $x_2$ when $x_2 \in [\frac{A-1}{2A} - x_1, \frac{A-1}{2A})$. Indeed, we have
\begin{align*}
\frac{\partial F(x_1, x_2)}{\partial x_2} &= \alpha A e^{A(x_2 - x_1)} \bigl(A(1 - x_1 - x_2) - 1\bigr)
\geq \alpha A e^{A(x_2 - x_1)} \Bigl(A\Bigl(1 - \frac{1}{A} - \frac{A-1}{2A}\Bigr) - 1\Bigr) \\
&= \alpha A e^{A(x_2 - x_1)} \cdot \frac{A - 3}{2} > 0.
\end{align*}
We conclude that
\begin{align*}
F(x_1, x_2) &\geq F\left(x_1, \frac{A-1}{2A} - x_1\right) = \left.\left(\alpha\bigl(A e^{A(\tilde x_2-x_1)} (1 - x_1-\tilde x_2) - 1\bigr) - (1 - Ax_1)\right)\right|_{\tilde x_2 = \frac{A-1}{2A} - x_1} \\
&\geq \alpha \cdot A\cdot \frac{1}{\alpha}
\cdot \left(1 - \frac{A-1}{2A}\right) - \alpha - (1-Ax_1) = \frac{A + 1}{2} - \alpha -1 +Ax_1 \geq \frac{A + 1}{2} - 2 > 0.
\end{align*}

This concludes the case analysis and the proof of Theorem~\ref{thm:mainBipartite}.
\end{proof} 
\section{Integrality Gap}

In this section, we give a $\Theta(\log 1/\alpha)$ integrality gap example for the LP relaxation presented in Section~\ref{sec:LP}. Notice that in the example each positive edge has a weight of $\bfw^+$ and each negative edge has a weight of $\bfw^-$ with $\bfw^+ \geq \bfw^-$.
\begin{proof}[Proof of Theorem~\ref{thm:intGap}]
Consider a $3$-regular expander $G = (V, E)$ on $n=\Theta((\alpha^2\log^2 \alpha)^{-1})$ vertices.
We say that two vertices $u$ and $v$ are similar if $(u, v) \in E$;
otherwise $u$ and $v$ are dissimilar. That is, the set of positive
edges $E^+$ is $E$ and the set of negative edges $E^-$ is $V\times V \setminus E$. Let $\bfw^+=1$ and $\bfw^-=\alpha.$

\begin{lemma}\label{lemma:gap-computation}
The integrality gap of the Correlation Clustering instance $G_{cc} = (V,E^+, E^-)$ described
above is $\Theta(\log \nicefrac{1}{\alpha})$.
\end{lemma}
\begin{proof}
Let $d(u,v)$ be the shortest path distance in $G$. Let $\varepsilon = 2/\log_3 n$. We define a feasible metric LP solution as follows:
$x_{uv} = \min (\varepsilon d(u,v), 1).$

Let $LP^+$ be the $LP$ cost of positive edges, and $LP^-$ be the LP cost of negative edges.
The LP cost of every positive edge is $\varepsilon$ since $d(u,v) = 1$ for $(u,v)\in E$.
There are $\nicefrac{3n}{2}$ positive edges in $G_{cc}$. Thus, $LP_+  < 3n/\log_3 n$. We now estimate $LP^-$. For every
vertex $u$, the number of vertices $v$ at distance less than $t$ is upper bounded by $3^t$ because $G$ is a 3-regular graph. Thus,
the number of vertices $v$ at distance less than $\nicefrac{1}{2} \log_3 n$ is upper bounded
by $\sqrt{n}$. Observe that the $LP$ cost of a negative edge $(u,v)$ (which is equal to $\alpha(1-x_{uv})$) is positive if and
only if $d(u,v) < \nicefrac{1}{2} \log_3 n$. Therefore, the number of negative edges with a positive $LP$ cost
incident on any vertex $u$ is at most $\sqrt{n}$. Consequently, the LP cost of all negative
edges is upper bounded by $\alpha n^{\frac{3}{2}}=\Theta(n/\log\nicefrac{1}{\alpha})$. Hence,
$$LP \leq \Theta(n/\log\nicefrac{1}{\alpha}) + 3n/\log_3 n = \Theta(n/\log\nicefrac{1}{\alpha}).$$
Here, we used that $\log n = \Theta(\log\nicefrac{1}{\alpha})$.

We now lower bound the cost of the optimal (integral) solution. Consider an optimal solution. There are two possible cases.
\begin{enumerate}
  \item No cluster contains 90\% of the vertices. Then a constant fraction of positive edges in the expander $G$ are cut and, therefore, the cost of the
  optimal clustering is at least $\Theta(n)$.
  \item One of the clusters contains at least 90\% of all vertices. Then all negative edges in that cluster are in disagreement with the clustering.
  There are at least  $\binom{0.9n}{2} - m = \Theta(n^2)$  such edges. Their cost is at least $\Omega(\alpha n^2)$.
\end{enumerate}

We conclude that the cost of the optimal solution is at least $\Theta(n)$ and, thus, the integrality gap is $\Theta(\log (1/\alpha))$.
\end{proof}

We note that in this example $\log (1/\alpha) = \Theta(\log n)$. However, it is easy to construct an integrality gap example where $\log (1/\alpha) \ll \Theta(\log n)$. To do so, we pick the integrality gap example constructed above and create $k\gg n$ disjoint copies of it. To make the graph complete, we add negative edges with (fractional) LP value equal to $1$ to connect each copy to every other copy of the graph. The new graph has $kn \gg n$ vertices. However, the integrality gap remains the same, $\Theta(\log \nicefrac{1}{\alpha})$.
\end{proof}
Now we give a $\Theta(\log 1/\alpha)$ integrality gap example when $G$ is a complete bipartite graph.
\begin{proof}[Proof of Theorem~\ref{thm:intGapBipartite}]
The proof is very similar to that of Theorem~\ref{thm:intGap}. We start with a 3-regular \textit{bipartite} expander $G = (L, R, E)$ on $n=\Theta((\alpha^2\log^2 \alpha)^{-1})$ vertices (e.g., we can use a 3-regular bipartite Ramanujan expander constructed by~\citet*{MSS13}). Then we
define a correlation clustering instance as follows: $G_{cc}=(L,R,E^+,E^-)$ where $E^+=E$ and $E^-=(L\times R) \setminus E$; let $\bfw^+=1$ and $\bfw^-=\alpha$.
The proof of Lemma~\ref{lemma:gap-computation} can be applied to $G_{cc}$; we only need to note that if
a cluster contains at least 90\% of the vertices, then there are at least $\Theta(n^2)$ edges of $G_{cc}$ between vertices in the cluster.
It follows that the integrality gap is $\Omega(\log(1/\alpha))$.
\end{proof}

\bibliography{corr-clust}
\bibliographystyle{plainnat}
\appendix 
\section{Proof of Theorem~\ref{thm:L4}}

For the sake of completeness we include the proof of Theorem~\ref{thm:L4} (see \cite{ACN08} and \cite{CMSY15}).

\begin{proof}[Proof of Theorem~\ref{thm:L4}] Our first task is to express the cost of violations made by Algorithm~\ref{alg:ApprAlg} and the LP weight in terms of $ALG^\sigma(\cdot)$ and $LP^\sigma(\cdot)$, respectively. In order to do this, we consider the cost of violations made by the algorithm at each step.

Consider step $t$ of the algorithm. Let $V_t$ denote the set of active (yet unclustered) vertices at the start of step $t$. Let $w\in V_t$ denote the pivot chosen at step $t$.
The algorithm chooses a set $S_t \subseteq V_t$ as a cluster and removes it from the graph. Notice that
for each $u \in S_t$, the constraint imposed by each edge of type $(u,v) \in E^+ \cup E^-$ is satisfied or violated right after step $t$. Specifically, if $(u,v)$ is a positive edge, then the constraint $(u,v)$ is violated if exactly one of the vertices $u,v$ is in $S_t$. If $(u,v)$ is a negative constraint, then it is violated if both $u,v$ are in $S_t$. Denote the weight of violated constraints at step $t$ by $ALG_t$. Thus,
\begin{align*}
ALG_t=&\sum\limits_{\substack{(u,v) \in E^+\\u,v \in V_t}} \bfw_{uv}\cdot\ONE \left(u \in S_t, v \not\in S_t \mbox{ or }u \not\in S_t, v \in S_t\right)+\sum\limits_{\substack{(u,v) \in E^-\\u,v \in V_t}} \bfw_{uv}\cdot\ONE \left(u \in S_t, v\in S_t\right).
\end{align*}

Similarly, we can quantify the LP weight removed by the algorithm at step $t$, which we denote by $LP_t$. 
We count the contribution of all edges $(u,v) \in E^+ \cup E^-$ such that $u \in S_t$ or $v \in S_t$. Thus,
\begin{align*}
LP_t =& \sum_{\substack{(u,v) \in E^+\\u,v \in V_t}} \bfw_{uv} x_{uv} \cdot \ONE(u \in S_t \text{ or } v \in S_t)+ \sum_{\substack{(u,v) \in E^-\\u,v \in V_t}} \bfw_{uv} (1 - x_{uv}) \cdot \ONE(u \in S_t \text{ or } v \in S_t)
\end{align*}

Note that the cost of the solution produced by the algorithm is the sum of the violations across all steps, that is $ALG = \sum_t ALG_t$. Moreover, as every edge is removed exactly once from the graph, we can see that $LP = \sum_t LP_t$. We will charge the cost of the violations of the algorithm at step $t$, $ALG_t$, to the LP weight removed at step $t$, $LP_t$. Hence, if we show that $\mathbb{E}[ALG_t] \leq \rho \mathbb{E}[LP_t]$ for every step $t$, then we can conclude that the approximation factor of the algorithm is at most $\rho$, since
$$
\mathbb{E}[ALG] = \mathbb{E}\bigg[\sum_t ALG_t \bigg] \leq \rho \cdot \mathbb{E}\bigg[\sum_t LP_t \bigg] =  \rho \cdot LP.
$$

We now express $ALG_t$ and $LP_t$ in terms of $cost(\cdot)$ and $lp(\cdot)$ which are defined in Section~\ref{triple_Analysis}. This will allow us to group together the terms for each triplet $u,v,w$ in the set of active vertices and thus write $ALG_t$ and $LP_t$ in terms of $ALG^\sigma(\cdot)$ and $LP^\sigma(\cdot)$, respectively.

For analysis, we assume that for each vertex $u \in V$, there is a positive (similar) self-loop, and thus we can define $cost(u,u\given w)$ and $lp(u,u \given w)$ formally as follows:
$cost(u,u \given w) = \pr(u \in S, u \not\in S \given p = w)
= 0$ and $lp(u,u \given w) = x_{uu} \cdot \pr(u \in S\given p = w) = 0$ (recall that $x_{uu}=0$).
\begin{align}\label{ALGcost}
\mathbb{E}[ALG_t \given  V_t&] = \sum_{\substack{(u,v) \in E\\u,v \in V_t}} \bigg( \frac{1}{|V_t|} \sum_{w \in V_t} \bfw_{uv} \cdot cost(u,v\given w)\bigg)= \frac{1}{2|V_t|}\sum_{\substack{u,v,w \in V_t \\ u \neq v}} \bfw_{uv}\cdot cost(u,v\given w)
\end{align}

\begin{align}\label{LPcost}
\mathbb{E}[LP_t \given V_t] &= \sum_{\substack{(u,v) \in E\\u,v \in V_t}} \bigg( \frac{1}{|V_t|} \sum_{w \in V_t} \bfw_{uv} \cdot lp(u,v\given w)\bigg)= \frac{1}{2|V_t|}\sum_{\substack{u,v,w \in V_t \\ u \neq v}} \bfw_{uv} \cdot lp(u,v\given w)
\end{align}

We divide the expressions on the right hand side by $2$ because the terms $cost(u,v \given w)$ and $lp(u,v \given w)$ are counted twice. Now adding the contribution of terms $cost(u,u \given w)$ and $lp(u,u \given w)$ (both equal to $0$) to~(\ref{ALGcost}) and~(\ref{LPcost}), respectively and grouping the terms containing $u,v$ and $w$ together, we get,
\begin{align*}
\mathbb{E}[ALG_t \given V_t] =& \frac{1}{6|V_t|}\sum_{u,v,w \in V_t}\bigg( \bfw_{uv}\cdot cost(u,v\given w)+ \bfw_{uw}\cdot cost(u,w\given v) + \bfw_{wv}\cdot cost(w,v\given u)\bigg)\\
=&\frac{1}{6|V_t|}\sum_{u,v,w \in V_t} ALG^\sigma(x,y,z)
\end{align*}
and
\begin{align*}
\mathbb{E}[LP_t \given V_t] =& \frac{1}{6|V_t|}\sum_{u,v,w \in V_t}\bigg(\bfw_{uv}\cdot lp(u,v\given w)+ \bfw_{uw}\cdot lp(u,w\given v) + \bfw_{wv}\cdot lp(w,v\given u)\bigg)\\
=&\frac{1}{6|V_t|}\sum_{u,v,w \in V_t}LP^\sigma(x,y,z)
\end{align*}

Thus, if $ALG^\sigma(x,y,z) \leq \rho LP^\sigma(x,y,z)$ for all signatures and edge lengths $x,y,z$ satisfying the triangle inequality, then $\mathbb{E}[ALG_t \given V_t] \leq \rho \cdot \mathbb{E}[LP_t \given V_t]$, and, hence, $\mathbb{E}[ALG] \leq \rho \cdot \mathbb{E}[LP]$ which finishes the proof.
\end{proof}

\end{document}